\ifx\GenSoCGVer\undefined%
\documentclass[12pt]{article}%
\newcommand{\RegVer}[1]{#1}
\newcommand{\SoCGVer}[1]{}
\else%

\documentclass[english]{socg/socg-lipics-v2019}%
\newcommand{\RegVer}[1]{}
\newcommand{\SoCGVer}[1]{#1}
\fi
\newcommand{\Undefine}[1]{\let#1\undefined}

\IfFileExists{sariel_computer.sty}{\def\sarielComp{1}}{}
\ifx\sarielComp\undefined%
\newcommand{\SarielComp}[1]{}
\newcommand{\NotSarielComp}[1]{#1}%
\else
\newcommand{\SarielComp}[1]{#1}%
\newcommand{\NotSarielComp}[1]{}%
\fi
\newcommand{\IfPrinterVer}[2]{#2}%

\RegVer{%
   \usepackage[cm]{fullpage}%
}%
\usepackage{amsmath}%
\usepackage{amssymb}%
\usepackage{xcolor}%

\SarielComp{\usepackage{sariel_colors}}%

\RegVer{%
   \usepackage[amsmath,thmmarks]{ntheorem}%
   \theoremseparator{.}%
}%

\usepackage{titlesec}%
\titlelabel{\thetitle. }%
\usepackage{xcolor}%
\usepackage{mleftright}%
\usepackage{xspace}%
\usepackage{hyperref}%

\providecommand{\BibLatexMode}[1]{}
\providecommand{\BibTexMode}[1]{#1}

\ifx\UseBibLatex\undefined%
  \renewcommand{\BibLatexMode}[1]{}
  \renewcommand{\BibTexMode}[1]{#1}
\else
  \renewcommand{\BibLatexMode}[1]{#1}
  \renewcommand{\BibTexMode}[1]{}
\fi

   \usepackage{euscript}%
\usepackage{graphicx}
\RegVer{\usepackage{caption}}%

\newcommand{\hrefb}[3][black]{\href{#2}{\color{#1}{#3}}}%

\IfPrinterVer{%
   \usepackage{hyperref}%
}{%
   \usepackage{hyperref}%
   \hypersetup{%
      breaklinks,%
      ocgcolorlinks, colorlinks=true,%
      urlcolor=[rgb]{0.25,0.0,0.0},%
      linkcolor=[rgb]{0.5,0.0,0.0},%
      citecolor=[rgb]{0,0.2,0.445},%
      filecolor=[rgb]{0,0,0.4},
      anchorcolor=[rgb]={0.0,0.1,0.2}%
   }
}

\SoCGVer{%
   \theoremstyle{remark}%
   \newtheorem{problem}[theorem]{Problem}%
   \newtheorem{defn}[theorem]{Definition}%
   \newtheorem{assumption}[theorem]{Assumption}%
   \newtheorem{observation}[theorem]{Observation}%

   \newtheorem*{remark:unnumbered}{Remark}%
}%

\RegVer{%
\theoremseparator{.}%

\theoremstyle{plain}%
\newtheorem{theorem}{Theorem}[section]

\newtheorem{lemma}[theorem]{Lemma}

\newtheorem{corollary}[theorem]{Corollary}

\newtheorem{observation}[theorem]{Observation}

\theoremstyle{plain}%
\theoremheaderfont{\sf}%
\theorembodyfont{\upshape}%
\newtheorem*{remark:unnumbered}[FakeCounter]{Remark}%

\newtheorem{defn}[theorem]{Definition}

\newcommand{\myqedsymbol}{\rule{2mm}{2mm}}

\theoremheaderfont{\em}%
\theorembodyfont{\upshape}%
\theoremstyle{nonumberplain}%
\theoremseparator{}%
\theoremsymbol{\myqedsymbol}%
\newtheorem{proof}{Proof:}%

}

\newcommand{\atgen}{\symbol{'100}}%
\newcommand{\SarielThanks}[1]{%
   \thanks{Department of Computer Science; %
      University of Illinois; %
      201 N. Goodwin Avenue; %
      Urbana, IL, 61801, USA; %
      {\tt sariel\atgen{}illinois.edu}; {\tt
         \url{http://sarielhp.org/}.} #1}}

\newcommand{\HLink}[2]{\hyperref[#2]{#1~\ref*{#2}}}
\newcommand{\HLinkSuffix}[3]{\hyperref[#2]{#1\ref*{#2}{#3}}}

\newcommand{\figlab}[1]{\label{fig:#1}}
\newcommand{\figref}[1]{\HLink{Figure}{fig:#1}}

\newcommand{\thmlab}[1]{{\label{theo:#1}}}
\newcommand{\thmref}[1]{\HLink{Theorem}{theo:#1}}

\newcommand{\seclab}[1]{\label{sec:#1}}
\newcommand{\secref}[1]{\HLink{Section}{sec:#1}}

\newcommand{\lemlab}[1]{\label{lemma:#1}}
\newcommand{\lemref}[1]{\HLink{Lemma}{lemma:#1}}%

\providecommand{\eqlab}[1]{}%
\renewcommand{\eqlab}[1]{\label{equation:#1}}

\providecommand{\deflab}[1]{\label{def:#1}}
\newcommand{\defref}[1]{\HLink{Definition}{def:#1}}

\providecommand{\remove}[1]{}%

\newcommand{\pth}[2][\!]{\mleft({#2}\mright)}%

\newcommand{\brc}[1]{\left\{ {#1} \right\}}
\newcommand{\cardin}[1]{\left| {#1} \right|}%

\renewcommand{\Re}{\mathbb{R}}%
\usepackage[inline]{enumitem}

\newlist{compactenumA}{enumerate}{5}%
\setlist[compactenumA]{topsep=0pt,itemsep=-1ex,partopsep=1ex,parsep=1ex,%
   label=(\Alph*)}%
\SoCGVer{%
   \setlist[compactenumA]{topsep=0pt,itemsep=-1ex,partopsep=1ex,parsep=1ex,%
      leftmargin=1cm,%
      label=(\Alph*)}%
}%

\newlist{compactenuma}{enumerate}{5}%
\setlist[compactenuma]{topsep=0pt,itemsep=-1ex,partopsep=1ex,parsep=1ex,%
   label=(\alph*)}%
\SoCGVer{%
   \setlist[compactenuma]{topsep=0pt,itemsep=-1ex,partopsep=1ex,parsep=1ex,%
      leftmargin=1cm,
      label=(\alph*)}%
}

\newlist{compactenumI}{enumerate}{5}%
\setlist[compactenumI]{topsep=0pt,itemsep=-1ex,partopsep=1ex,parsep=1ex,%
   label=(\Roman*)}%
\SoCGVer{%
   \setlist[compactenumI]{topsep=0pt,itemsep=-1ex,partopsep=1ex,parsep=1ex,%
      leftmargin=1cm, %
   label=(\Roman*)}%
}

\newlist{compactenumi}{enumerate}{5}%
\setlist[compactenumi]{topsep=0pt,itemsep=-1ex,partopsep=1ex,parsep=1ex,%
   label=(\roman*)}%
\SoCGVer{%
\setlist[compactenumi]{topsep=0pt,itemsep=-1ex,partopsep=1ex,parsep=1ex,%
   leftmargin=1cm, label=(\roman*)}%
}

\newlist{compactitem}{itemize}{5}%
\setlist[compactitem]{topsep=0pt,itemsep=-1ex,partopsep=1ex,parsep=1ex,%
   label=\bullet}%

\numberwithin{figure}{section}%
\numberwithin{table}{section}%
\numberwithin{equation}{section}%

\providecommand{\Mh}[1]{#1}%

\providecommand{\Mh}[1]{#1}%
\providecommand{\pth}[1]{({#1})}%
\newcommand{\EdgesX}[1]{\Mh{E}\pth{#1}}%
\newcommand{\edgeX}[1]{\Mh{e}\pth{#1}}%
\newcommand{\eminX}[1]{\Mh{e_{\min}}\pth{#1}}%
\newcommand{\emaxX}[1]{\Mh{e_{\max}}\pth{#1}}%
\newcommand{\spanX}[1]{\Mh{\mathrm{span}}\pth{#1}}%
\newcommand{\IAX}[1]{\Mh{\mathrm{in}}\pth{#1}}%
\newcommand{\OAX}[1]{\Mh{\mathrm{out}}\pth{#1}}%

\newcommand{\diffX}[1]{\Mh{\nabla}\pth{#1}}%

\usepackage[normalem]{ulem}

\newcommand{\CHX}[1]{{\mathcal{CH}}\pth{#1}}

\newcommand{\CPoly}{\Mh{\pi}}%
\newcommand{\Poly}{\Mh{\sigma}}%

\newcommand{\Line}{\Mh{\ell}}%
\newcommand{\PP}{\Mh{P}}%

\newcommand{\QP}{\Mh{\EuScript{Q}}}%
\newcommand{\IP}{\Mh{\EuScript{I}}}%
\newcommand{\UP}{\Mh{\EuScript{U}}}%
\newcommand{\CP}{\Mh{\EuScript{C}}}%

\newcommand{\SegSet}{\Mh{S}}%

\newcommand{\PPart}{\Mh{\mathcal{P}}}%

\newcommand{\CNF}{\Mh{\mathcal{C}}}%
\newcommand{\ICover}{\Mh{\Pi}}%
\newcommand{\PSet}{D}%

\newcommand{\Forest}{\Mh{\mathcal{F}}}%

\providecommand{\emphi}[1]{\index{#1}\textcolor{blue}{\bf{\emph{#1}}}}
\newcommand{\Caratheodory}{Carath\'eodory\xspace}
\newcommand{\areaX}[1]{\mathrm{area}\pth{#1}}%

\newcommand{\cF}{\Mh{F}}%
\newcommand{\cG}{\Mh{G}}%

\newcommand{\pa}{\Mh{p}}%

\newcommand{\pb}{\Mh{q}}%
\newcommand{\PB}{\Mh{Q}}%

\newcommand{\etal}{\textit{et~al.}\xspace}

\newcommand{\joinop}{\texttt{join}\xspace}
\newcommand{\splitop}{\texttt{split}\xspace}

\newcommand{\edge}{{e}}
\newcommand{\ray}{{\rho}}

\SoCGVer{%
   \newcommand{\myparagraph}[1]{%
      \noindent%
      \textbf{#1} %
   }
}
\RegVer{%
   \newcommand{\myparagraph}[1]{\paragraph{#1}}%
}

\newcommand{\vorigX}[1]{v_{#1}}%
\newcommand{\currX}[1]{\mathrm{curr}\pth{#1}}%
\newcommand{\queue}{\Mh{\mathcal{Q}}}%

\newcommand{\labelX}[1]{\mathrm{label}\pth{#1}}

\newcommand{\DSCH}{\mathcal{D}_{\mathrm{ch}}}
\newcommand{\DSUF}{\mathcal{D}_{\mathrm{uf}}}
\newcommand{\DSRay}{\mathcal{D}_{\mathrm{ray}}}
\newcommand{\polysX}[1]{\IP(#1)}%
\newcommand{\seg}{\Mh{s}}%

\begin{document}

\title{Covering Polygons by Min-Area Convex Polygons}%

\RegVer{%
   \author{%
      Elias Dahlhaus%
      \and%
      Sariel Har-Peled%
      \SarielThanks{%
         Work on this paper was partially supported by a NSF AF award
         CCF-1907400. %
      }%
      \and%
      Alan L. Hu%
   }%
}

\maketitle

\begin{abstract}
    Given a set of disjoint simple polygons
    $\Poly_1, \ldots, \Poly_n$, of total complexity $N$, consider a
    convexification process that repeatedly replaces a polygon by its
    convex hull, and any two (by now convex) polygons that intersect
    by their common convex hull. This process continues until no pair
    of polygons intersect.

    We show that this process has a unique output, which is a cover of
    the input polygons by a set of disjoint convex polygons, of total
    minimum area.  Furthermore, we present a near linear time
    algorithm for computing this partition. The more general problem
    of covering a set of $N$ segments (not necessarily disjoint) by
    min-area disjoint convex polygons can also be computed in near linear time.

    A similar result is already known, see the work by
    Barba \etal
    \cite{bbbs-ccpf-13}.
\end{abstract}

\section{Introduction}

Let $\Poly_1, \ldots, \Poly_n$ be $n$ simple disjoint polygons in the
plane with a total of $N$ vertices (a polygon is simple, if it has no
holes).  We would like to break them into maximal number of groups of
polygons, such that each group can be separated from any other groups
by a line.  This partition can also be interpreted as computing the
minimum area coverage of the input polygons by disjoint convex
polygons.

\begin{figure}[h]
    \begin{tabular}{c|c|c}
      \includegraphics[page=1,width=0.3\linewidth]{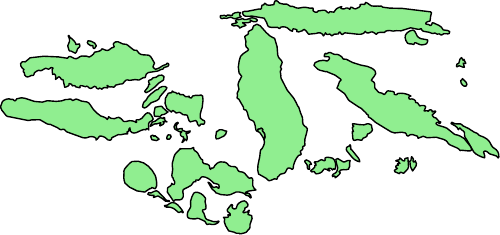}%
      &
        \includegraphics[page=2,width=0.3\linewidth]{figs/solomon_islands}%
      &
        \includegraphics[page=3,width=0.3\linewidth]{figs/solomon_islands}
      \\%
      \hline
      &&\\
      \includegraphics[page=4,width=0.3\linewidth]{figs/solomon_islands}
      &
        \includegraphics[page=6,width=0.3\linewidth]{figs/solomon_islands}
      &
        \includegraphics[page=7,width=0.3\linewidth]{figs/solomon_islands}
    \end{tabular}
    \caption{Convexification in action.}
    \figlab{example}
\end{figure}

Specifically, a \emph{convex cover} is a minimal set of disjoint
convex polygons that cover the input polygons.  We are interested in
computing the convex cover that has the following two equivalent
properties: (i) minimizes the total area of the convex polygons, or
(ii) maximizes the total number of polygons $m$ in the cover,

\SoCGVer{\bigskip}%
\myparagraph{Convexification.} %
The desired partition can be by computed by a \emph{convexification}
process. Specifically, given a finite set $\IP$ of disjoint simple
polygons in the plane, we start by replacing each polygon of $\IP$
with its convex hull. Next, each pair of (now convex) polygons of
$\IP$ that intersect, are replaced by the convex hull of their
union. The process repeats until no pair of polygons intersect.  See
\figref{example}.

As we show below, this process has a unique well defined output, and
it provides the desired partition of the input polygons.

\SoCGVer{\bigskip}%
\myparagraph{The challenge.} %
The challenge is providing a fast implementation of the
convexification process.  The natural approach is to try and do divide
and conquer. However, since there are inputs for this problem where
the merge process requires a linear number of sequential pairs of
polygons to be merged, this approach would not work directly. The
sequential nature of this process is illustrated in \figref{seq}.

\begin{figure}[h]
    \includegraphics[page=1]{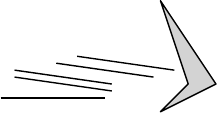} $\implies$
    \includegraphics[page=2]{figs/slow} $\implies$
    \includegraphics[page=3]{figs/slow} $\implies$
    \includegraphics[page=4]{figs/slow}
    
    \hfill $\implies$ \includegraphics[page=5]{figs/slow} $\implies$
    \includegraphics[page=6]{figs/slow} $\implies$
    \includegraphics[page=7]{figs/slow}

    \caption{The sequential nature of the convexification process.}
    \figlab{seq}
\end{figure}
An alternative approach is to try and use a ray shooting
data-structure, but such data-structures are too expensive, since ray
shooting on general disjoint polygons requires $\Omega( n^{4/3} )$
time if one performs a linear number of ray shooting queries, because
such a data-structure can be used to solve Hopcroft's problem. It is
not a priori clear that near linear time algorithm is possible for
this problem.

\SoCGVer{\bigskip}%
\myparagraph{Our results.} %
We show that the convexification of disjoint simple polygons in the
plane, of total complexity $N$, can be computed in $O(N \log^2 N)$
time. If the polygons are not disjoint, the running time becomes
$O(N \alpha(N) \log^2 N)$, where $\alpha$ is the inverse Ackermann
function.

To this end, we prove that the convexification process is well
defined, and as stated above, has a unique result. Next, we use a
data-structure of Ishaque \etal \cite{ist-sprdp-12} to perform ray
shooting to decide if the convex-hulls of polygons intersect. This
data-structure inserts the rays that it shoots into the scene, thus
avoiding the pitfalls of the standard ray-shooting data-structures,
resulting in near linear time for the queries performed. To keep track
of the convex-hulls as they are being merged, we introduce a
data-structure for dynamic maintenance of convex-hulls that behaves
like pseudo-disks (which is an invariant of the algorithm).

Because the convexification process is somewhat inconstant, the
analysis and the algorithm requires some care.

\SoCGVer{\bigskip}%
\myparagraph{Previous work.} %
Barba \etal \cite{bbbs-ccpf-13} derived a very similar result for a
collection of disjoint trees in the plane.

\SoCGVer{\bigskip}%
\myparagraph{Paper organization.} %
We describe the algorithm in \secref{algorithm}.  In \secref{c:good}
we prove that convexification is well defined, and prove some basic
properties we would need to analyze the algorithm.  The analysis of
the algorithm itself is in \secref{analysis}. In \secref{result} we
present the main result, and show how to extend it to a general set of
segments.  We describe the data-structure for maintaining
convex-hulls, under intersection and merge operations, in
\secref{c:h:datastructure}.

\section{Algorithm}
\seclab{algorithm}

The input is a set $\IP$ of $n$ simple polygons that are disjoint,
with a total of $N$ vertices. We assume that all the vertices of $\IP$
are in general position (i.e., no three of them lie on a common line).

\subsection{Data-structures used by the algorithm}
\subsubsection{A data-structure for maintaining convex-hulls} %
\seclab{dynamic:c:h}

We need a data-structure for maintaining a set of convex polygons that
supports merge and intersection detection.  It is well known how to
maintain convex-hulls under insertion of points, with $O( \log n)$
time per insertion, see \cite[Section 3.3.6]{ps-cgi-85}. We need a
slightly more flexible data-structure that supports also intersection
detection, similar in spirit to the data-structure of Dobkin and
Kirkpatrick \cite{dk-dsppu-90}.

\begin{lemma}
    \lemlab{data_structure}%
    The input is a set of convex polygons with total complexity
    $N$. One can preprocess them in linear time, such that the
    following operations are supported.
    \begin{compactenumI}
        \smallskip%
        \item Decide, in $O( \log N)$ time, if a query point is inside
        a specified polygon in the set.
        
        \smallskip%
        \item Compute, in $O( \log N)$ time, if two specified convex
        polygons in the set intersect, and if so return a point in
        their common intersection.
        
        \smallskip%
        \item Compute, in $O( \log N)$ time, the segment of
        intersection between a specified convex polygon in the set,
        and a query line.
        
        \smallskip%
        \item Given two convex-polygons $\CPoly_1$ and $\CPoly_2$ in
        the set, such that their boundaries intersect at most twice,
        compute, in $O( (1+u)\log N)$ time, the convex polygon
        $\CPoly = \CHX{\CPoly_1\cup\CPoly_2}$, which replaces
        $\CPoly_1$ and $\CPoly_2$ in the set of polygons. Here, $u$ is
        the number of vertices of $\CPoly_1$ and $\CPoly_2$ that do
        not appear in $\CPoly$ (i.e., the number of vertices that were
        deleted).
    \end{compactenumI}
\end{lemma}
Building the above data-structure is relatively a straightforward
modification of known techniques, and it is described in
\secref{c:h:datastructure}.

\subsubsection{Ray shooting data-structure.}

Ray shooting on a set of polygons is expensive in general, since using
$n$ ray shootings one can solve Hopcroft's problem, which is believed
\cite{e-nlbhp-96} to require $\Omega( n^{4/3})$ time.  Fortunately, we
can use a ray-shooting data-structure of Ishaque \etal
\cite{ist-sprdp-12} -- this data-structure shoots permanent rays, that
are added to scene after they are being shot.  When shooting $O(N)$
rays, using this data-structure, on an initial scene made out of $n$
polygons with total complexity $N$, the total running time of this
data-structure is $O( N \log^2 N)$ time.

\subsection{Algorithm in detail}

The algorithm is illustrated in \figref{alg:exec}.

\SoCGVer{\bigskip}%
\myparagraph{Initialization.} %
The algorithm initializes the ray shooting data-structure $\DSRay$
described above for the input polygons.  Next, the algorithm computes
the convex-hull of each of the input polygons. All these convex
polygons are stored in the data-structure $\DSCH$ of
\lemref{data_structure} for maintaining convex-hulls. All the new
edges introduced when computing the convex-hulls are stored in a FIFO
queue $\queue$. In addition, the algorithm initializes a union-find
data-structure $\DSUF$, where each input polygon is an element. A set
in the union-find data-structure $\DSUF$ represents a set of input
polygons that were fused together into a larger convex polygon. This
convex-polygon would be one the convex polygons maintained by
$\DSCH$. Every input polygon $\Poly$ sets its label to be itself; that
is, $\forall \Poly \in \IP$, $\labelX{\Poly} = \Poly$. %

\newcommand{\miniG}[1]{%
   \begin{minipage}{0.3\linewidth}
       \smallskip%
       #1%
       \medskip
   \end{minipage}%
}%
\begin{figure}[p]
    \hfill%
    \begin{tabular}{c|c|c}  
  {\includegraphics[page=1]{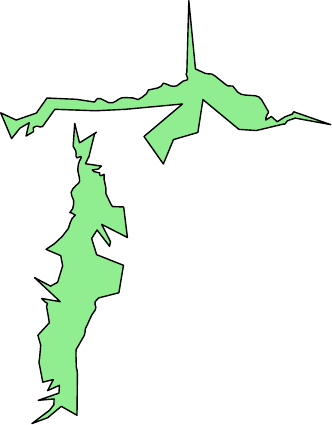}}
  &
    \includegraphics[page=2]{figs/algorithm}
  &
    {\includegraphics[page=3]{figs/algorithm}}
  \\
  \miniG{The input polygons.}
  &
    \miniG{  The convex-hulls of the input polygons.}
  & \miniG{The edges in the initial queue.}\\
  \hline
  \hline%
  ~&~&
  \\
  \includegraphics[page=4]{figs/algorithm}
  &
    \includegraphics[page=5]{figs/algorithm}
  \\
  \miniG{Ray shooting till the first collision. The two new blue edges
  are added to the queue.}
  &
    \miniG{The final set of ray shootings performed.}
  &
    \end{tabular}%
    \hfill\phantom{}%
\caption{Illustration of the algorithm execution.}
\figlab{alg:exec}
\end{figure}

\SoCGVer{\bigskip}%
\myparagraph{Execution.} %
Using the ray-shooting data-structure $\DSRay$ the algorithm traces
the edges stored in $\queue$, in a FIFO fashion.  Specifically, the
algorithms pops the edge $\edge$ in the front of $\queue$. The
endpoints $\pa_1$ and $\pa_2$ of $\edge$ belong two input polygons
$\pa_1 \in \Poly_1$ and $\pa_2 \in\Poly_2$, respectively (it is
possible that $\Poly_1 = \Poly_2$). The algorithm shoots a ray $\ray$
from $\pa_1$ towards $\pa_2$ using $\DSRay$ (here
$\labelX{\ray} = \Poly_1$).  The ray $\ray$ must hit something, and
let $\Poly$ be the entity being hit (if the ray arrives to $\pa_2$,
then $\Poly=\Poly_2$). Here $\Poly$ might be an original input
polygon, or a segment that is the trace of an older ray shooting.

The algorithm computes, using $\DSUF$, the two sets $X$ and $X'$ that
contains $\labelX{\Poly_1}$ and $\labelX{\Poly}$.  If $X \neq X'$,
then the ray $\ray$ hit a different connected component than its own.
Next the algorithm retrieves the two convex polygons $\CPoly$ and
$\CPoly'$ in $\DSCH$ that corresponds to $X$ and $X'$, respectively.
The algorithm replaces $\CPoly$ and $\CPoly'$ in $\DSCH$ by
$\CHX{ \CPoly \cup \CPoly'}$. The new convex-hull has two new edges,
and these two edges are pushed into $\queue$. The algorithm merges $X$
and $X'$ into a larger set in $\DSUF$.

The algorithm now continues to the next edge in $\queue$. This
stage ends when the queue $\queue$ is empty.

\SoCGVer{\bigskip}%
\myparagraph{Cleanup stage.} %
The above results in a collection of convex polygons $\CP$ all with
disjoint boundaries. It is still possible that some polygons are
contained inside some other polygons, but this can be readily handled
in $O( N \log N)$ time by doing sweeping -- which would remove all the
redundant inner polygons.

\section{Properties of convexification}
\seclab{c:good} \seclab{forest}

The algorithm execution results in a convexification of the input
polygons. However, the basic convexification process itself is not
uniquely defined, and there are many different executions of a
convexification process. As such, here we prove that the
convexification process always results in the same set of convex
polygons.

\subsection{Preliminaries}

In the following, $\IP$ denotes the input set of $n$ disjoint simple
polygons.  For a set of polygons $\UP$, let
$\CHX{\UP} = \CHX{ \cup_{\Poly \in \UP}^{} \Poly}$ be the combined
convex hull of all the polygons of $\UP$.

\paragraph*{Basic properties of convex-hulls}

\begin{lemma}
    For any two sets $X, Y \subseteq \Re^2$, we have
    \begin{math}
        \CHX{ X \cup Y } = \CHX{ \CHX{X} \cup \CHX{Y} \bigr.}.
    \end{math}
\end{lemma}
\begin{proof}
    Let $\cF = \CHX{ X \cup Y }$ and
    $\cG = \CHX{ \CHX{X} \cup \CHX{Y} \bigr.}$.  Clearly,
    $\cF \subseteq \cG$. As for the other direction, consider any
    point $\pa \in \cG$. By \Caratheodory theorem, there exists a set
    $\PB = \{\pb_1, \pb_2. \pb_3\} \subseteq \CHX{X} \cup \CHX{Y}$,
    such that $\pa \in \CHX{ \PB \bigr.}$. Applying \Caratheodory
    theorem again, we have that there are sets
    $\PB_i\subseteq X \cup Y$, such that $\pb_i \in \CHX{\PB_i}$, for
    $i=1,\ldots,3$. As such,
    $\pa \in \CHX{ \PB \bigr.} \subseteq \CHX{ \PB_1 \cup \PB_2 \cup
       \PB_3} \subseteq \CHX{ X \cup Y}$. This implies that
    $\pa \in \cF$, and as such $\cG \subseteq \cF$.
\end{proof}

\paragraph*{The convexification process} %

\begin{defn}
    \deflab{convexification} %
    Given a set of polygons $\IP$, consider a process that starts with
    a partition $\PPart_0$ of $\IP$, where every polygon of $\IP$ is a
    singleton. A \emphi{convexification} $\CNF$ of $\IP$ is now a
    sequence of a partitions
    $\PPart_0, \PPart_1, \PPart_2, ..., \PPart_{t}$. For $i>0$, we
    have
    \begin{compactenumi}
        \smallskip%
        \item there are two sets of polygons
        $\PSet_{i-1},\PSet_{i-1}' \in \PPart_{i-1}$ such that
        $\CHX{\PSet_{i-1}}$ and $\CHX{\PSet_{i-1}'}$ intersect (i.e.,
        $\PSet_{i-1}$ and $\PSet_{i-1}'$ are not separable).
        
        \smallskip%
        \item
        $\PPart_i = \pth{\PPart_{i-1} \setminus \{ \PSet_{i-1},
           \PSet_{i-1}'\}} \cup \brc{\CHX{ \PSet_{i-1} \cup
              \PSet_{i-1}'}}$.
    \end{compactenumi}
    \smallskip%
    Furthermore, for any two sets of polygons
    $\PSet_{i}, \PSet_{j} \in \PPart_{t}$, we have $\CHX{\PSet_{i}}$
    and $\CHX{\PSet_{j}}$ do not intersect.
\end{defn}

\paragraph*{Representing convexification as a forest}

Given a convexification $\CNF$ of $\IP$, it can be represented as a
forest $\Forest$ of binary reverse trees (similar in spirit, but not
in details, to the trees used internally by the disjoint union-find
data-structure).  Initially, each input polygon $\Poly \in \IP$ has a
leaf node $u$ that it corresponds to, such that
$\polysX{u} = \{\Poly\}$.

Each internal node $v$ of $\Forest$ corresponds to a subset of
$\polysX{v} \subseteq \IP$ that was created by the convexification
process. Specifically, if two subsets $\polysX{x}$ and $\polysX{y}$
were merged by the process to form $\polysX{v}$, then $x$ and $y$ are
the two children of $v$ in the forest, and
$\polysX{v} = \polysX{x} \cup \polysX{y}$.

A set $\UP \subseteq \IP$ is a \emphi{connected component} of a forest
$\Forest$, if there is a root node $u$ of a tree in $\Forest$, such
that $\polysX{u} = \UP$.

\subsection{Convexifications results in unique %
   cover by convex polygons}

\begin{lemma}
    \lemlab{the:same}%
    Let $\IP$ be a set of disjoint simple polygons.  The final
    partition of a convexification $\CNF$ of $\IP$ is unique. In other
    words, the set of convex polygons forming $\CNF$ depends only on
    the input polygons in $\IP$, and not on the order in which the
    convexification was performed.
\end{lemma}

\begin{proof}
    Consider two convexifications of $\IP$, $\CNF$ and $\CNF'$, with
    merge forests $\Forest$ and $\Forest'$, respectively. We claim
    that the connected components of $\Forest$ and $\Forest'$ are the
    same.
    
    We prove by induction that for any node $u$ of $\Forest$, there
    exists a node of $u'$ in $\Forest'$, such that
    $\polysX{u} \subseteq \polysX{u'}$.  When $|\polysX{u}| = 1$, the
    claim is immediate, as all polygons appear as singletons in the
    beginning of the process, and there is a leaf of $u'$ of
    $\Forest'$ that stores the same polygon as $u$.

    So, assume inductively that this holds for all nodes $u$ in
    $\Forest$ with $\cardin{\polysX{u}} < N$. Now consider a node $x$
    in $\Forest$ such that $\cardin{\polysX{u}} = N > 1$. It has two
    children $y,z$ in the forest $\Forest$, such that
    $\polysX{u} = \polysX{x} \cup \polysX{y}$ and
    $\CHX{\polysX{x}} \cap \CHX{\polysX{y}} \neq \emptyset$.
    Furthermore, $\cardin{\polysX{x}} < N$ and
    $\cardin{\polysX{y}} < N$. By induction, there are nodes $x', y'$
    in $\Forest'$, such that $\polysX{x} \subseteq \polysX{x'}$ and
    $\polysX{y} \subseteq \polysX{y'}$.  As such,
    $\CHX{ \polysX{x} } \subseteq \CHX{ \polysX{x'} }$ and
    $\CHX{ \polysX{y} } \subseteq \CHX{ \polysX{y'} }$. As such,
    \begin{equation*}
        \CHX{\polysX{x'}} \cap \CHX{\polysX{y'}}%
        \supseteq%
        \CHX{\polysX{x}} \cap \CHX{\polysX{y}}%
        \neq \emptyset.
    \end{equation*}
    But this implies that $x'$ and $y'$ must belong to the same tree
    of $\Forest'$. Let $u'$ be the root of their common tree in
    $\Forest'$, and observe that
    \begin{math}
        \polysX{u'}%
        \supseteq%
        \polysX{x'} \cup \polysX{y'}%
        \supseteq%
        \polysX{x} \cup \polysX{y} = \polysX{u}.%
    \end{math}
    Implying the claim.
    
    Applying the claim symmetrically, we get that any connected
    component of $\Forest$ is contained in a connected component of
    $\Forest'$, and vice versa. This implies that the connected
    components of the two forests are the same, and as such their
    convexifications are the same.
\end{proof}

\subsection{More properties of convexification}

\begin{defn}
    A set $\UP$ of disjoint simple polygons in the plane is
    \emphi{separable} if there exists a line $\ell$ which does not
    pass through any polygon in $\UP$, and there are polygons on both
    sides of $\ell$. Otherwise, the set $\UP$ is \emphi{tight} -- that
    is, any line $\ell$ either intersects a polygon of $\IP$, or all
    polygons of $\IP$ lie to one side of $\ell$.
\end{defn}

\begin{lemma}
    \lemlab{tight:line} %
    Consider a tight set $\QP$ of disjoint polygons. A line $\Line$
    intersects $\CHX{\QP}$ $\iff$ $\Line$ intersects some polygon in
    $\QP$.
\end{lemma}

\begin{proof}
    If a line $\Line$ intersects some polygon in $\QP$ then it
    definitely intersects the larger set $\CHX{\QP}$.  As for the
    other direction, if there is a line $\Line$ that intersects
    $\CHX{\QP}$ but none of the polygons of $\QP$, then $\Line$
    separates $\QP$, which is a contradiction.
\end{proof}

\begin{lemma}
    \lemlab{tight:union}%
    If $\IP_1, \IP_2 \subseteq \IP$ are tight, and $\CHX{\IP_1}$
    intersects $\CHX{\IP_2}$, then $\IP_1 \cup \IP_2$ is tight.
\end{lemma}

\begin{proof}
    Any line that separates $\IP_1 \cup \IP_2$, either (i) separates
    $\IP_1$, (ii) separates $\IP_2$, or (iii) separates $\IP_1$ from
    $\IP_2$. All there possibilities are impossible by assumption.
\end{proof}

\begin{lemma}
    Given a convexification $\CNF$ with a forest $\Forest$, for any
    node $u$ of $\Forest$, we have that $\polysX{u}$ is tight.
\end{lemma}

\begin{proof}
    Follows readily from \lemref{tight:union}.
\end{proof}

\subsection{The result}

\begin{defn}
    \deflab{convex:cover}%
    A \emphi{convex cover} of a set of polygons $\IP$ is a set of
    disjoint convex polygons $\CPoly_1, \ldots, \CPoly_m$, such that
    each polygon $\Poly \in \IP$ is contained in some polygon
    $\CPoly_{j(\Poly)}$. Furthermore, every polygon $\CPoly_i$, in the
    cover, contains at least one of the polygons of $\IP$.
\end{defn}

The \emphi{area} of a convex cover $\ICover$ is
$\areaX{ \ICover}= \sum_{\CPoly \in \ICover} \areaX{\CPoly}$.

\begin{theorem}
    Given a set of disjoint simple polygons $\IP$, any convexification
    of $\IP$, results in a unique convex cover $\CNF$ of
    $\IP$. Furthermore, we have
    \begin{compactenumi}
        \item $\CNF$ is the minimum area convex cover of $\IP$, and 

        \item $\CNF$ is the convex cover of $\IP$ of maximum
        cardinality.
    \end{compactenumi}
\end{theorem}
\begin{proof}
    The uniqueness follows readily from \lemref{the:same}.
    
    Consider a tight set of polygons $\UP \subseteq \IP$. Any convex
    cover $\ICover$ of $\IP$, must have a single convex polygon
    $\Poly$, such that $\CHX{\UP} \subseteq \Poly$. Since, each
    polygon of $\CNF$ is the convex hull of a tight set, it follows
    that each of the polygons of $\CNF$ are contained in some polygon
    of $\ICover$, which readily implies the above.
\end{proof}

\section{Analysis of the algorithm}
\seclab{analysis}

\subsection{On the partitions maintained by the %
   algorithm}

Here, we define and prove some invariant properties of the partition
of the input polygons maintained by the algorithm.  The key property
needed for the algorithm to work is that for any two sets of polygons that
are being merged together, their corresponding convex-hulls behave
like pseudo-disks. See \lemref{compatible:intersection} below for
details.

\begin{defn}
    Consider a set $\QP \subseteq \IP$ of polygons, and a set
    $\SegSet$ of interior disjoint close segments. The set $\QP$ is
    \emphi{bridgeable} by $\SegSet$, if
    \begin{compactenumi}
        \item $\bigcup \QP \cup \bigcup \SegSet$ is a connected,
        \item the segments of $\SegSet$ and the polygons of $\IP$ are
        interior disjoint,

        \item all the polygons of $\IP$ that the segments of $\SegSet$
        intersect belong to $\QP$,
               
        \item all segments of $\SegSet$ have at least one endpoint on
        the boundary of a polygon in $\QP$, and

        \item all segments of $\SegSet$ lie within $\CHX{\QP}$.
    \end{compactenumi}
    See \figref{bridgeable} for an example.
\end{defn}

\begin{figure}[h]
    \centerline{\includegraphics{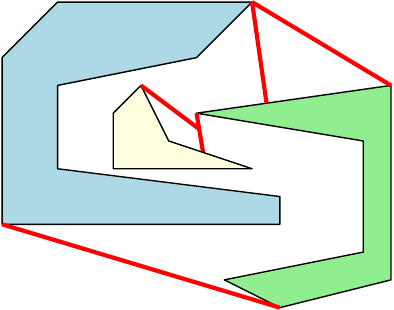}}%
    \caption{A bridgeable set.}
    \figlab{bridgeable}
\end{figure}

\begin{defn}
    \deflab{proper} %
    A partition $\PPart$ of $\IP$ is \emphi{proper} if all the sets in
    $\PPart$ are tight and bridgeable.
\end{defn}

\begin{observation}
    Let $\IP$ be a set of disjoint polygons in the plane.  The initial
    partition $\PPart$ of $\IP$, where each polygon of $\IP$ is a
    singleton, is proper.
\end{observation}

\begin{defn}
    \deflab{compatible}%
    Two bridgeable sets $\PP_1, \PP_2 \subseteq \IP$ are
    \emphi{compatible} if
    \begin{compactenumi}
        \item $\PP_1$ and $\PP_2$ are disjoint,
        \item there are sets of segments $\SegSet_1$ and $\SegSet_2$,
        such that $\PP_1$ and $\PP_2$ are bridgeable by $\SegSet_1$
        and $\SegSet_2$, respectively,
        \item the segments of $\SegSet_1$ are disjoint from the
        segments of $\SegSet_2$, and
        \item $\PP_1 \cup \PP_2$ is bridgeable.
    \end{compactenumi}
\end{defn}

\begin{figure}[h]
    \includegraphics[page=1]{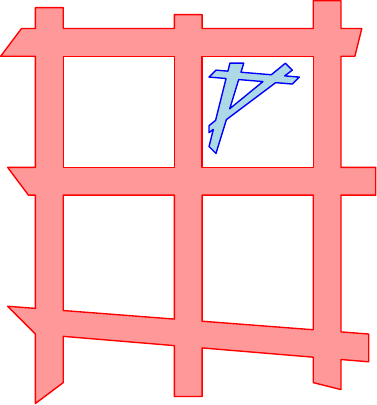}
    \hfill%
    \includegraphics[page=2]{figs/two_polygons}
    \hfill%
    \includegraphics[page=3]{figs/two_polygons}
    \captionof{figure}{}
    \figlab{c:i}
\end{figure}

\begin{lemma}
    \lemlab{compatible:intersection}%
    Let $\PP_1, \PP_2 \subseteq \IP$ be two sets of polygons that are
    compatible, then the boundaries of $\CHX{\PP}$ and $\CHX{\QP}$
    intersects at most twice.
\end{lemma}

\begin{proof}
    Let $\SegSet_1$ and $\SegSet_2$ be the two compatible sets of
    bridges for $\PP_1$ and $\PP_2$, respectively.  Consider the two
    polygons formed by $\PP_1 \cup \SegSet_1$ and
    $\PP_2 \cup \SegSet_2$ (they are not necessarily simple).  It is
    well known that the convex-hulls of two disjoint simple polygons
    behaves like pseudo-disks (see for example \cite[Lemma
    3.2]{hks-akltd-17}). But this holds also for two disjoint polygons
    that are not necessarily simple. Indeed replace each polygon by
    the polygon formed by its outer boundary. The interesting case is
    when one of the resulting polygons is fully contained inside the
    other, but then their boundaries do not intersect at all. See
    \figref{c:i}.
\end{proof}

\subsection{Correctness}

Here, we show that the algorithm indeed computes the convexification
of the input polygons. In particular, we show that during the main
stage of the algorithm, each merge is between two compatible subsets.

In the following, consider the forest $\Forest$, described in
\secref{forest}, that the algorithm (implicitly) maintains over the
input polygons.

For a ray $\ray$, shot by the algorithm, let $\vorigX{\ray}$ be the
node in $\Forest$ that gave rise to it. When the ray $\ray$ is finally
being issued by the algorithm, $\vorigX{\ray}$ might be an internal
node, and let $\currX{\vorigX{\ray}}$ be the current root of the tree
containing $\vorigX{\ray}$ in $\Forest$.

Each connected component of the scene at any point time corresponds to
a set of polygons that is bridgeable -- indeed, the bridges being the
segments formed by the ray shootings.  In particular, the polygons
together with the segments formed by the rays shot so far, partition
the polygons into connected components, where each root node in
$\Forest$ corresponds to one of the resulting connected components.

\begin{observation}
    Consider a segment $\seg$ that was the result of a ray shooting query,
    rising out of a node $v$ of $\Forest$.  Specifically, $\seg$ is
    contained in an edge $\edge$ of $\CHX{ \polysX{v}}$ which was
    created when two convex polygons where merged. The segment $\seg$
    was created by a ray shooting along $\edge$.  We then have that
    $\seg \subseteq \CHX{ \polysX{v}} \subseteq \CHX{
       \polysX{\currX{v}}}$.
\end{observation}

\begin{lemma}
    \lemlab{inv:2}%
    Assume the algorithm computed a sequence of partitions
    $\PPart_0, \PPart_1, ..., \PPart_{m}$, where $\PSet_t, \PSet_{t}'$
    are the two sets of polygons of $\PPart_{t}$ that the algorithm
    merged at time $t$, for $t=0,\ldots, m - 1$. Then, for any $t$, we
    have:
    \begin{compactenumA}
        \smallskip%
        \item $\CHX{\PSet_t}$ intersects $\CHX{\PSet_{t}'}$.
        \smallskip%
        \item $\PSet_t$ and $\PSet_{t}'$ are compatible
        (\defref{compatible}).  \smallskip%
        \item $\PPart_{t + 1}$ is proper -- that is, all the sets of
        $\PPart_{t+1}$ are tight and bridgeable.
    \end{compactenumA}    
\end{lemma}

\begin{proof}
    Recall that $\PPart_0$ is proper.  So, assume the claim holds,
    inductively, for $0 \le t < n$, and consider $t=n$.  If the
    algorithm merged $\PSet$ and $\PSet'$, at time $n$, then a ray
    shot $\ray$ originating from a vertex of a polygon in $\PSet$, hit
    a polygon in $\PSet'$, or it hit a previously inserted segment
    $\seg$, such that $\polysX{\currX{\vorigX{\seg}}} = \PSet'$.  The
    resulting segment satisfies $\ray \subseteq \CHX{\PSet}$, and
    $\ray \cap \CHX{\PSet'} \neq \emptyset$, which implies that
    $\CHX{\PSet}$ intersects $\CHX{\PSet'}$ (i.e., it is a justified
    merge). This readily implies, by induction, that
    $\PSet \cup \PSet'$ is tight.

    By induction, $\PSet$ and $\PSet'$ are bridged by sets of segments
    $\SegSet$ and $\SegSet'$, respectively. The segments of $\SegSet$
    and $\SegSet'$ are interior disjoint,
    $\SegSet \subseteq \CHX{\PSet}$, and
    $\SegSet' \subseteq \CHX{\PSet'}$. The interior of the segment/ray
    $\ray$ does not intersect the set
    $\PSet \cup \PSet' \cup \SegSet \cup \SegSet'$. Furthermore, one
    endpoint of $\ray$ belongs to $\PSet$, and the other belongs to
    $\PSet' \cup \SegSet'$. Thus, $\PSet \cup\PSet$ is bridgeable by
    $\SegSet \cup \SegSet' \cup \{\ray\}$, and $\PSet$ and $\PSet'$
    are compatible.
\end{proof}

\begin{lemma}
    \lemlab{proof:convexification}%
    Consider the partition $\PPart_m$ computed by the algorithm (just
    before the beginning of the second stage). Then, for any two sets
    $\PP_1, \PP_2 \in \PPart_m$, we have that
    $\partial \CHX{\PP_1} \cap \partial \CHX{\PP_2} = \emptyset$.
\end{lemma}

\begin{proof}
    For the sake of contradiction, assume that there are disjoint sets
    $\PP_1, \PP_2 \in \PPart_{m}$, such that there are edges
    $\edge_1 \in \partial \CHX{\PP_1}$ and
    $\edge_2 \in \partial \CHX{\PP_2}$ that intersect. Since the input
    polygons are disjoint, at least one of these edges, say $\edge_1$
    does not belong to the original polygons. But then, the algorithm
    performed a ray shooting query along it. Since $\PP_2$ is tight,
    by \lemref{inv:2}, it follows that this ray shooting must hit some
    original polygon, that is not in $\PP_1$. But that would readily
    imply that the algorithm would further enlarge $\PP_1$ to be a
    bigger connected set, by merging the connected component being hit
    into $\PP_1$, but this is a contradiction to $\PP_1$ being a set
    in the final partition computed by the first stage of the
    algorithm.
\end{proof}

\begin{lemma}
    The algorithm computes a convexification of the original input
    polygons.
\end{lemma}
\begin{proof}    
    In the end of the first stage, there might be two components
    $\PP_1$ and $\PP_2$ in the computed partition, such that
    $\CHX{\PP_1} \subseteq \CHX{\PP_2}$. This case is discovered by
    the sweeping, which merges $\PP_1$ into $\PP_2$. This together
    with \lemref{proof:convexification} implies that the resulting
    partition is indeed a convexification. Note, that by
    \lemref{inv:2}, every convex polygon being output is the
    convex-hull of a tight set of input polygons, and no two
    convex-polygons intersect, which implies that this is indeed the
    desired convexification. 
\end{proof}

\subsection{Running time}

\subsubsection{Preprocessing}%
The ray-shooting data structure, described by Ishaque \etal
\cite{ist-sprdp-12}, requires $O(N \log N)$ time and space to
construct. For a simple polygon with $n$ vertices, computing its
convex hull takes $O(n)$ time, so doing this for all polygons in $\IP$
takes $O(N)$ time overall. And this is also the time it takes to
compute the edges of the convex hulls which do not belong to the
original polygons and inserting them into queue $\queue$. Building the
union-find data structure also takes $O(N)$ time and space. Thus,
preprocessing takes a total of $O(N \log N)$ time and space.

\subsubsection{Main stage}%
For each edge in our queue $\queue$, algorithm shoots a ray, and
potentially unions two sets of polygons and computes their joined
convex hull. By \lemref{compatible:intersection}, the number of new
edges inserted into $\queue$ for each merge is at most 2. Thus, the
total number of edges to be inserted into $\queue$ is $O(N)$ as
well. For $N$ initial points and $O(N)$ rays, the ray-shooting data
structure takes $O(N \log^2 N)$ time and $O(N \log N)$ space.

Computing the union of two sets in a union-find data structure takes
amortized $O(\alpha (N))$ time \cite{t-cawrn-79}, where $\alpha$ is
the inverse Ackermann function. Over all merges, this becomes
$O(N \alpha (N))$ time. As described in \lemref{data_structure},
computing the joined convex hull takes $O((1 + u) \log N)$ time, where
$u$ is the number of deleted vertices. Since each vertex can be
deleted at most once, the sum of this over all merges is
$O(N \log N)$. Thus, the total time during the main stage of algorithm
is $O(N \log^2 N)$, with the bottleneck being the ray-shooting data
structure.

\SoCGVer{\bigskip}%
\myparagraph{Overall.} %
The plane sweep in the second stage takes $O(N \log N)$ time.  Thus,
the total running time is $O(N \log^2 N)$, and the algorithm uses
$O(N \log N)$ space.
\section{The result and an extension}
\seclab{result}

Putting the above together, we get the following result.

\begin{theorem}
    \thmlab{main}%
    Given a set $\IP$ of simple disjoint polygons in the plane, with
    total complexity $N$, one can compute, in $O( N \log^2 N)$ time,
    the convexification of $\IP$. The resulting set of polygons, is a
    cover of the polygons of $\IP$ by convex-polygons, of total
    minimum area.
\end{theorem}

\begin{corollary}
    Given a set $\IP$ of $N$ segments the plane, one can compute, in
    $O( N \alpha(N) \log^2 N)$ time, the convexification of $\IP$,
    where $\alpha$ is the inverse Ackermann function. The resulting
    set of polygons, is a cover of the segments of $\IP$ by
    convex-polygons, of total minimum area.
\end{corollary}
\begin{proof}
    Compute the outer face of the arrangement of the segments of
    $\IP$. This takes $O( N \alpha(N) \log^2 N)$ time
    \cite{sa-dsstg-95}. Every connected component of the boundary of
    this face, can be interpreted now as a simple polygon, and the
    total complexity of these polygons is $O( N \alpha(N))$
    \cite{sa-dsstg-95}.  The task at hand, is to compute the
    convexification of these polygons, which can be done in
    $O( N \alpha(N) \log^2 N)$ time, by \thmref{main}.
\end{proof}

\section{Data structure for dynamic maintenance %
   of convex-hull}
\seclab{c:h:datastructure}

Here, we describe how to build the data-structure of
\lemref{data_structure} -- we emphasize that the resulting
data-structure is a relatively easy variant of known results, and the
detailed description is included here for the sake of completeness.

\subsection{Representing a convex polygon}
We maintain the convex-hull of each polygon as two lists of edges for
the top and bottom chains, respectively. Each chain is stored from
left to right, using a balanced binary search tree that supports
insertions, and deletion.  In addition, we need the \splitop operation
-- it break such a (sorted) list into two sorted lists, that starts
and ends at a specific object. Similarly, we need the \joinop
operation, which merges two sorted lists (where one chain is to the
left of other).

In addition, we need in-order successor/predecessor queries in
constant time. We augment the tree, such that every internal node
(which stores an edge), also stores the first and last edges stored in
this subtree.

A specific implementation of a balanced binary search tree that
provides the desired properties is a red-black tree -- all operations
can be performed in $O( \log n)$ time, except for the
successor/predecessor operations which takes $O(1)$ time.

\subsection{Deciding if a point is inside a polygon}

This readily follows by doing a binary search on the top and bottom
chains to find the edges intersecting the vertical line through the
query point. This clearly takes $O( \log N)$ time.

\subsection{Deciding if two polygons intersect}

Every polygon is represented using two chains -- the algorithm checks
all four possible combination if they intersect.

The algorithm for checking if two polygons, each represented by a
chain, intersects is recursive.  Initially, the chain is represented
by the root node of the tree. Thus, the recursive intersection checker
is given two nodes $u,v$ in the two respective trees representing the
chains, and the task is to decide if the two subpolygons of $u$ and
$v$ intersect.

Specifically, for a node $v$, let $\EdgesX{v}$ be all the edges of the
chains stored in the subtree of $v$. The subpolygon associated with
$v$, is $\Poly_v = \CHX{v}$. As such, the task is to decide if
$\CHX{ v}$ and $\CHX{u}$ intersects. Let $\eminX{v}$ and $\emaxX{v}$
be the leftmost and rightmost edges stored in $\EdgesX{v}$,
respectively. Let $\edgeX{v}$ be the edge stored in $v$.  Let
$\spanX{v}$ be the segment connecting the two $x$-extreme vertices of
$\eminX{v}$ and $\emaxX{v}$. See \figref{inner:outer}.

\begin{figure}
    \hfill%
    \includegraphics[page=1]{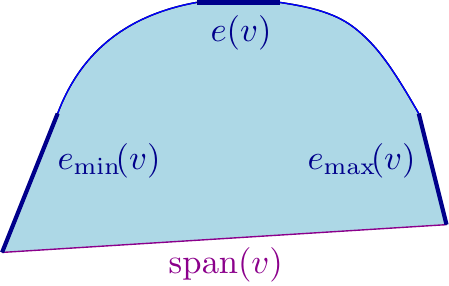} \hfill%
    \includegraphics[page=2]{figs/inner_outer}%
    \hfill%
    \phantom{}
    \caption{Inner and outer approximation of the portion of the
       convex body stored at a vertex $v$.}%
    \figlab{inner:outer}
\end{figure}
   
In constant time, one can compute the \emphi{inner} approximation
\begin{equation*}
    \IAX{v} = \CHX{ \edgeX{v} \cup \eminX{v} \cup \emaxX{v} \bigr.}    
\end{equation*}
to $\Poly_v$. The \emphi{outer} approximation $\OAX{v}$ is the
intersection of the four halfspaces containing $\IAX{v}$, with their
boundary lines passing through the four edges
$\edgeX{v}, \eminX{v}, \emaxX{v}, \spanX{v}$. Clearly, the outer
approximation can also be computed in constant time.

If the two inner approximations $\IAX{u}$ and $\IAX{v}$ intersect,
then the algorithm returns the two polygons intersect and return an
intersection point.  If the two outer approximations $\OAX{u}$ and
$\OAX{v}$ do not intersect, then the algorithm returns that there is
no intersection.

The set $\diffX{v} = \OAX{v} \setminus \IAX{v}$ is made out of two
triangles.  If the two inner approximations do not intersect, then
there is a line $\Line$ that separates them, and this line intersects
only one triangle of $\diffX{v}$. If this triangle intersects
$\IAX{u}$, then the algorithm continues the search recursively on the
child of $v$ that corresponds to this triangle, and $u$. See
\figref{ear}. Similarly, if one of the triangles of $\diffX{u}$
intersects $\IAX{v}$, then the algorithm continues recursively on the
appropriate child of $u$, and $v$.

\begin{figure}[h]
    \centerline{\includegraphics{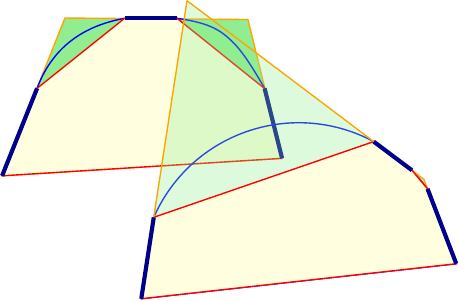}}%
    \caption{An ear of an outer approximations, intersects the inner
       approximation of the other body.}
    \figlab{ear}
\end{figure}

So, the only situation that remains is that one triangle of
$\diffX{v}$ intersects one triangle of $\diffX{u}$. The algorithm
continues recursively the search for intersection in the two
respective children of $u$ and $v$.

Since the depth of the two trees is logarithmic, it follows that an
intersection point, if it exists, would be found in $O( \log N)$ time.

\subsection{Computing the intersection points %
   of a polygon with a line}
\seclab{i:points}

We are given a query line $\Line$, and a convex polygon $\CPoly$. We
are interested in computing the two endpoints of the segment
$\CPoly \cap \Line$.  As before, the algorithm computes the
intersections with the top and bottom parts separately.  Given a node
$v$, the algorithm continues recursively into a triangle of
$\diffX{v}$ $\iff$ if it intersects $\Line$. This results in (at most)
two paths in the tree, which can be computed in $O( \log N)$. The two
edges the leaves of these paths corresponds two, contains the two
endpoints, which can be readily computed.

\subsection{Computing the convex-hull of two %
   intersecting convex polygons}

We are given two convex polygons $\CPoly_1$ and $\CPoly_2$ that
intersect. Importantly, the two polygons intersect as pseudo-disks --
their boundaries intersect at most twice. The task at hand is to
compute $\CHX{\CPoly_2 \cup \CPoly_2}$. To this end, the algorithm
computes a point $\pa \in \CPoly_1 \cap \CPoly_2$, using the algorithm
of \secref{i:points}.

\begin{figure}[h]
    \includegraphics[page=1,width=0.22\linewidth]{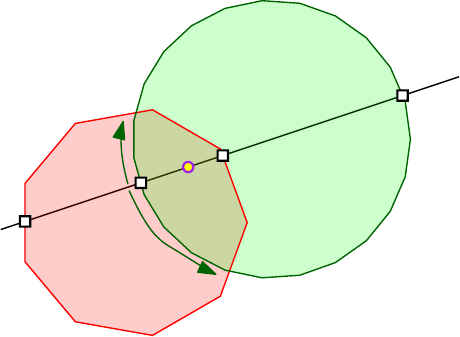} \hfill
    \includegraphics[page=2,width=0.22\linewidth]{figs/delete} \hfill
    \includegraphics[page=3,width=0.22\linewidth]{figs/delete} \hfill
    \includegraphics[page=4,width=0.22\linewidth]{figs/delete}
    \caption{Illustration of the algorithm for computing the
       convex-hull of two convex polygons.}
    \figlab{convex:intersections}%
\end{figure}

The algorithm is depicted in \figref{convex:intersections}.  Let
$\Line$ be the horizontal line though $\pa$, and consider the four
intersection points along $\Line$. The middle two intersections are
inside the other convex-hull. In particular, assume that one of these
intersections is $\pa_1$, and it lies in
$\partial \CPoly_1 \cap \CPoly_2$. The algorithm walks along the edges
of $\partial \CPoly_1$, in both directions, starting at $\pa_1$,
deleting edges if they lie completely inside $\Poly_2$. Checking if an
edge lies inside $\CPoly_2$ can be done in $O( \log N)$ time by
checking if its two endpoints are inside $\CPoly_2$. This takes
$O( \log N)$ time per deleted edge.

Once the algorithm arrives to an edge of $\partial \CPoly_1$ that
intersects $\partial \CPoly_2$, it computes the intersection point
(using, say, line intersection query), and again, we start working
along the portion of $\partial \CPoly_2$ that lies inside
$\CPoly_1$. This results in discovering the two intersection points of
the boundaries of $\CPoly_1$ and $\CPoly_2$. The algorithm thus
deleted all the edges of one polygon that lines inside another. Next,
the algorithm modifies the intersecting edges, so that they share the
intersection point as a common endpoint. The algorithm then performs a
join operation on the top and bottom chains of the two polygons, to
get the representation of the two polygons $\CPoly_1 \cup
\CPoly_2$. Naturally, this polygon is no longer convex, However, the
algorithm can now compute the two bridges between the two chains. To
this end, starting with the intersection vertex, the algorithm checks
if a vertex is a valley (i.e., concave), and if so, it removes it. It
keep doing this check till all the vertices in this vicinity are
convex. Clearly, this take time proportional the number of vertices
deleted. The algorithm does a similar process on the bottom
intersection point. Together, this results in the desired convex-hull,
in time $O( (1+u) \log N)$, where $u$ is the number of vertices of
$\CPoly_1$ and $\CPoly_2$ that are no longer on the boundary of the
final convex-hull.

\begin{remark:unnumbered}
    Note, that the property that the convex-hull boundaries intersect
    only in two points is critical in making the above algorithm
    work. If there are more intersections, then it is not even clear
    how to efficiently compute them in time proportional to the number
    of intersections.
\end{remark:unnumbered}

 \providecommand{\CNFX}[1]{ {\em{\textrm{(#1)}}}}
  \providecommand{\CNFCCCG}{\CNFX{CCCG}}


\begin{thebibliography}{BBBS13}

\bibitem[BBBS13]{bbbs-ccpf-13}
Luis Barba, Alexis Beingessner, Prosenjit Bose, and Michiel H.~M. Smid.
\newblock \href{http://cccg.ca/proceedings/2013/papers/paper\_18.pdf}{Computing
  covers of plane forests}.
\newblock In {\em Proc. 25th Canad. Conf. Comput. Geom.\CNFCCCG}. Carleton
  University, Ottawa, Canada, 2013.

\bibitem[DK90]{dk-dsppu-90}
D.~P. Dobkin and D.~G. Kirkpatrick.
\newblock \href{https://doi.org/10.1007/BFb0032047}{Determining the separation
  of preprocessed polyhedra -- a unified approach}.
\newblock In {\em Proc. 17th Int. Colloq. Automata Lang. Prog. {\em(ICALP)}},
  volume 443 of {\em Lect. Notes in Comp. Sci.}, pages 400--413.
  Springer-Verlag, 1990.

\bibitem[Eri96]{e-nlbhp-96}
\hrefb{http://compgeom.cs.uiuc.edu/~jeffe/}{J.~{Erickson}}.
\newblock \href{https://doi.org/10.1007/BF02712875}{New lower bounds for
  {Hopcroft's} problem}.
\newblock {\em \hrefb{http://link.springer.com/journal/454}{Discrete Comput. {}Geom.}}, 16(4):389--418, 1996.

\bibitem[HKS17]{hks-akltd-17}
\hrefb{http://sarielhp.org}{S.~{{Har-Peled}}}, \hrefb{http://www.cs.tau.ac.il/~haimk/}{H.~{Kaplan}}, and \hrefb{http://www.math.tau.ac.il/~michas}{M.~{Sharir}}.
\newblock \href{https://doi.org/10.1007/978-3-319-44479-6}{Approximating the
  $k$-level in three-dimensional plane arrangements}.
\newblock In M.~Loebl, J.~Ne{\v{s}}et{\v{r}}il, and R.~Thomas, editors, {\em A
  Journey Through Discrete Mathematics: A Tribute to Ji{\v r}{\' i} Matou{\v
  s}ek}, pages 467--503. Springer, 2017.

\bibitem[IST12]{ist-sprdp-12}
Mashhood Ishaque, Bettina Speckmann, and Csaba~D. T{\'{o}}th.
\newblock \href{https://doi.org/10.1137/100804310}{Shooting permanent rays
  among disjoint polygons in the plane}.
\newblock {\em {SIAM} J. Comput.}, 41(4):1005--1027, 2012.

\bibitem[PS85]{ps-cgi-85}
F.~P. Preparata and M.~I. Shamos.
\newblock \href{https://doi.org/10.1007/978-1-4612-1098-6}{{\em Computational
  Geometry: An Introduction}}.
\newblock Springer-Verlag, 1985.

\bibitem[SA95]{sa-dsstg-95}
\hrefb{http://www.math.tau.ac.il/~michas}{M.~{Sharir}} and \hrefb{http://www.cs.duke.edu/~pankaj}{P.~K.~{Agarwal}}.
\newblock
  \href{http://us.cambridge.org/titles/catalogue.asp?isbn=0521470250}{{\em
  {Davenport-Schinzel} Sequences and Their Geometric Applications}}.
\newblock Cambridge University Press, New York, 1995.

\bibitem[Tar79]{t-cawrn-79}
R.~E. Tarjan.
\newblock  A class of algorithms which require nonlinear time to maintain
  disjoint sets.
\newblock {\em J. Comput. Syst. Sci.}, 18:110--127, 1979.

\end{thebibliography}
\end{document}